\documentclass[11pt,a4papers]{article}
\pdfoutput=1
\usepackage{graphicx}
\usepackage{stmaryrd}
\usepackage{amssymb}
\usepackage{amsmath}
\usepackage{amsthm}
\usepackage{dsfont}
\usepackage{hyperref}
\usepackage{amsfonts}
\usepackage{amsrefs}

\newtheorem{theorem}{Theorem}
\newtheorem{lemma}[theorem]{Lemma}

\theoremstyle{definition}

\theoremstyle{remark}

\newcommand{\Loop}{\ell}
\newcommand{\ue}{\bar {e}}

\newcommand{\Id}{\mathrm{Id}}

\newcommand{\trace}{\text{tr}}
\newcommand{\FinG}{\mathcal{G}}

\newcommand{\Walk}{\omega}		
\newcommand{\Walks}{\mathcal{W}}

\newcommand{\Loops}{\mathcal{L}}

\newcommand{\uLoop}{\ell^{\circ}}
\newcommand{\polsim}{\sim}

\newcommand{\rw}{w}

\makeatletter
\newcommand{\subjclass}[2][2010]{%
  \let\@oldtitle\@title%
  \gdef\@title{\@oldtitle\footnotetext{#1 \emph{Mathematics Subject Classification.} #2}}%
}

\title{A short proof of the Kac--Ward formula}

\author{
Marcin Lis
 \thanks{Mathematical Sciences, Chalmers University of Technology and the University of Gothenburg, SE-412 96 G\"oteborg, Sweden; {e-mail:} \nolinkurl{marcinl@chalmers.se}}
 \thanks{ ICERM, Brown University, 121 South Main Street, RI 02903 Providence, USA; {e-mail:} \nolinkurl{marcin_lis@brown.edu} } 
 }
\date{}


\begin{document}

\maketitle

\begin{abstract}  We present a new short proof of the Kac--Ward formula for the partition function of the
Ising model on planar graphs.
\end{abstract}

Let $\FinG=(V,E)$ be a finite graph embedded in the complex plane with non-intersecting edges drawn as straight line segments.
For a directed edge $e=(t_e,h_e)$, its reversal is $-e=(h_e,t_e)$, and its undirected version is $\ue=\{ t_e,h_e\} \in E$.
For two directed edges $e,g$, the \emph{turning angle} from  $e$ to $g$ is \begin{align*} 
\angle(e,g)= \text{Arg}\Big(\frac{h_g-t_g}{h_e-t_e}\Big) \in (-\pi,\pi]
\end{align*}
(see Figure~\ref{fig}).
Let $x=(x_{\ue})_{\ue \in E}$ be a vector of real edge weights.
The \emph{transition matrix} is a matrix indexed by the directed edges and is given by
\begin{align*}
\Lambda_{e,g} = \begin{cases}
		x_{\ue} e^{\frac{i}{2}\angle(e,g)}
		& \text{if $h_{e} = t_{g}$ and $g \neq -e$}; \\
		0 & \text{otherwise}.
	\end{cases}
\end{align*}
An \emph{even subgraph} is a set $H\subset E$ such that the degree of each vertex of $(V,H)$ is even. Let
\[
Z = \sum_{H \text{ even}}  \ \prod_{\ue\in H}x_{\ue}
\]
be the generating function of even subgraphs, where the product over the empty set is taken to be $1$. 
If $x_{\ue} \in(0,1)$, then $Z$ is
the partition function of the Ising model~\cite{Ising} defined on $\FinG$. We refer the reader to~\cite{KLM} for more details
on the connection with the Ising model.
The main result of this note is a short proof of the following theorem.
\begin{theorem}[Kac--Ward formula]\label{thm}
\begin{align}\label{eq:KacWard}
{\det}(\Id -\Lambda) =Z^2,
\end{align}\
where $\Id$ is the identity matrix.
\end{theorem}

 \begin{figure}
		\begin{center}
			\includegraphics{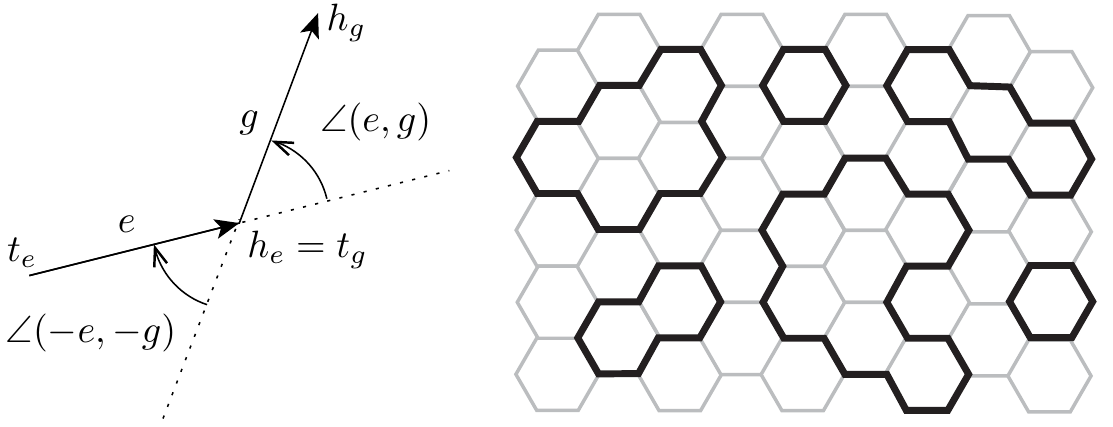}  
		\end{center}
	\caption{The turning angle and an even subgraph of the hexagonal lattice}
	\label{fig}
\end{figure}

Many papers appeared in the physics and mathematics literature where the Kac--Ward formula is proved or claimed to be proved. 
The original proof of Kac and Ward~\cite{KacWard} famously contained an error. Subsequently, several papers appeared where attempts were made to fix it.
We mention the contributions of Sherman~\cite{Sherman1}, Burgoyne~\cite{Burgoyne}, and Vdovichenko~\cite{Vdovichenko1},
where loop expansions of the determinant were used. However, these papers still left a lot to wish for in terms of mathematical rigour.
In the light of the accessible computation of the partition function due to Kasteleyn~\cite{Kasteleyn1}, who expressed it as the partition function of a dimer model on a decorated graph, the combinatorics involved in the Kac--Ward formula seemed unnecessarily  
complicated.
This was probably the reason why its first rigorous proof was given only much later by Dolbilin et al.\ \cite{DolEtAl}. 
In recent years more rigorous but still combinatorially involved proofs appeared~\cites{Cimasoni,KLM,Helmuth}.
We refer the reader to~\cite{KLM} for a longer discussion on the history of this theorem. We also need to mention that there exists a short proof due to Chelkak, Cimasoni and Kassel, who 
discovered it while investigating the double Ising model~(private communication). 

In this note a new short proof based on the loop expansion of the determinant is presented. Like all the previous proofs it relies on cancellations between certain weighted combinatorial objects.
In our case, these objects are \emph{loops}. The main improvement in comparison with~\cite{KLM} is that there is no need for expanding the generating functions
into generating functions of \emph{collections of loops}, which complicates the picture.
The cancellations of loop weights fall into two categories: generic and specific (Lemma~\ref{lem:generic} and Lemma~\ref{lem:specific} respectively).
The generic cancellations follow from the general theory of loop-erased walks. 
The specific cancellations are an easy consequence of the unique sign-changing property of the weights induced by the transition matrix.
The combinatorial mechanism of the Kac--Ward formula is therefore as transparent as the one of the loop-erased walk.

\section*{Walks and loops}
A \emph{(non-backtracking) walk} $\Walk$
of length $|\Walk|=n$ is a sequence of directed edges $\Walk=(\Walk_1,\ldots,\Walk_{n+1})$
 such that $t_{\Walk_{i+1}}=h_{\Walk_i}$ and $\Walk_{i+1} \neq -\Walk_i$ for $1\leq i\leq n$. Note that the length of a walk is the number of steps
 the walk makes between the edges rather than the number of edges itself.  By $\Walks_{e,g}$ we denote the set of all walks starting at $e$ and ending at $g$. 
 For $\Walk \in \Walks_{e,g}$ and $\Walk'\in \Walks_{g,h}$,
\[
\Walk \oplus \Walk'=(\Walk_1,\ldots,\Walk_{|\Walk|+1},\Walk'_2,\ldots, \Walk'_{|\Walk'|+1}) \in \Walks_{e,h}
\]
is the concatenation of $\Walk$ and $\Walk'$, and $\Walk^{-1}=(-\Walk_{|\Walk|},\ldots, -\Walk_1) \in \Walks_{-g,-e}$
is the reversal of $\Walk$. If $1\leq k \leq l \leq |\Walk|+1$, then $\Walk_{k,l} =(\Walk_k,\ldots,\Walk_l)$. \emph{Loops} are walks of length larger than $1$ and belonging to $\Walks_{e,e}$ for some directed edge $e$, and are denoted by $\Loop$. 
A loop $\Loop$ is \emph{self-avoiding} if each vertex appears in exactly two edges of $\Loop_{1,|\Loop|}$.
 
The transition matrix induces complex-valued weights on walks:
\[
\lambda(\Walk) = \prod_{i=1}^{|\Walk|}\Lambda_{\Walk_i,\Walk_{i+1}}= e^{\frac i2 \alpha(\Walk)} x(\Walk),
\]
where
\[
 x(\Walk)=\prod_{i=1}^{|\Walk|}x_{\overline{\Walk_i}}, \quad \text{and} \quad \alpha(\Walk)=\sum_{i=1}^{|\Walk|} \angle(\Walk_{i},\Walk_{i+1}).
\]
Note that the weight of the last edge is not included in the weight of the walk. 
The crucial properties of these weights are stated below.
\begin{lemma}
\begin{align}
 &\lambda( \Walk \oplus \Walk') =\lambda(\Walk)\lambda(\Walk') \qquad  &\text{for } \Walk \in \Walks_{e,g} \text{ and } \Walk' \in \Walks_{g,h}, \tag{\emph{i}}\\
 &\lambda(\Walk)=  -\lambda(\Walk^{-1})=\pm i x(\Walk) \quad\qquad &\text{for } \Walk \in \Walks_{e,-e}, \label{prop2} \tag{\emph{ii}} \\
&\lambda(\Loop) = \lambda(\Loop^{-1}) = \pm x(\Loop)\quad \qquad &\text{for a loop } \Loop \label{prop3} \tag{\emph{iii}} ,\\
&\lambda(\Loop) =- x(\Loop) \qquad &\text{for a self-avoiding loop } \Loop \label{prop4} \tag{\emph{iv}}.
\end{align}
\end{lemma}
\begin{proof}
Multiplicativity follows from the definition. To obtain~\eqref{prop2} and~\eqref{prop3}, recall that $\text{Arg}(z/w) =\text{Arg}(z)-\text{Arg}(w) \text{ (mod $2\pi $)}$ 
for any $z,w \neq 0$, and therefore
$\alpha(\Walk) = \angle(e,g) \ (\text{mod } 2\pi)$ for $\Walk \in \Walks_{e,g}$.  It is now enough to notice that $\alpha(\Walk)=-\alpha(\Walk^{-1})$ for any walk $\Walk$.
Finally, if $\Loop$ is self-avoiding, then $\alpha(\Loop)$ is the sum of the exterior angles of the polygon defined by~$\Loop$.
Hence, $\alpha(\Loop)=\pm 2\pi$, and~\eqref{prop4} follows.
\end{proof}
 
Before proving the theorem, we need to define a few more notions. A loop~$\Loop$ is \emph{rooted} at $e$ if $\Loop \in \Walks_{e,e}$.
The \emph{signed measure} of a loop $\Loop$  is given by
\[
\rw(\Loop) = \frac{\lambda(\Loop)}{|\Loop|}.
\]
\emph{Unrooted loops} are equivalence classes of loops under
the cyclic shift relation $\Loop\sim \Loop_{i,|\Loop|} \oplus \Loop_{1,i}$,
and are denoted by $\uLoop$. 
With a slight abuse of notation, if $f$ is a function defined on loops which is invariant under cyclic shifts, then $f(\uLoop)$ is the evaluation of $f$ at any representative of~$\uLoop$.

The \emph{multiplicity} of a loop~$\Loop$, denoted by $m_{\Loop}$, is the largest number $m$ such that $\Loop = (\Loop')^{\oplus m}$ for some loop $\Loop'$. We say that $\Loop$
\emph{visits} a directed edge~$e$ $k$ times if $e$ appears $k$ times in $\Loop_{1,|\Loop|}$. Note that for each edge $e$, the number of times $\Loop$ visits $e$ is always divisible by, but not necessarily equal to $m_{\Loop}$.

If $L$ is a set of loops, then we will write $\lambda(L)= \sum_{\Loop \in L} \lambda(\Loop)$ and $\rw(L) = \sum_{\Loop \in L} \rw(\Loop)$. Unnecessary brackets will be omitted in 
this notation, i.e.\ $w \{ \ldots \} = w (\{ \ldots \})$.
Note that since $L$ can be infinite,
it will always be assumed that $\|x\|_{\infty}=\max_{\ue \in E} |x_{\ue}|$ is sufficiently small to
guarantee that all such power series
are absolutely summable. This in particular implies that the order in which the sums are taken is irrelevant. Since the walks are non-backtracking, it is actually enough to take
$\|x\|_{\infty} <1/(\Delta-1)$, where $\Delta$ is the maximal degree of $\FinG$.

\section*{Proof of Theorem~\ref{thm}} \label{sec:KacWard}
Note that since both sides of~\eqref{eq:KacWard} are polynomials, it is enough to prove the desired equality for small $\| x \|_{\infty}$. 

The first two lemmas use only the fact that $\lambda$ is a multiplicative weight. 
\begin{lemma}[Loop expansion of the determinant]\label{lem:loopexpansion}  Let $\Loops$ be the set of all rooted loops. Then,
\begin{align*}
\det(\Id-\Lambda)=\exp\big(-\rw(\Loops)\big).
\end{align*}
\end{lemma}
\begin{proof}
Let $\alpha_i$ be the eigenvalues of $\Lambda$. Then,
\begin{align*}
\rw(\Loops)&= \sum_{n=1}^{\infty} \sum_{|\Loop|=n} \frac{\lambda(\Loop)}{n} = \sum_{n=1}^{\infty} \frac{\trace \Lambda^n}{n} = \sum_i \sum_{n=1}^{\infty} \frac{\alpha_i^n}{n} \\& =-\ln  \prod_i(1-\alpha_i) =  -\ln \det(\Id-\Lambda).	
\qedhere 
\end{align*}
\end{proof}

The next lemma is a variant of Lemma~9.3.2 of \cite{LawlerLimic}, which is used to prove the exponential formula for the law of the loop-erased walk.
\begin{lemma}[Generic cancellations] \label{lem:generic} 
Let $\Loops^1_e$ be the set of loops rooted at $e$ which visit $e$ only once and do not visit $-e$.
Then,
\begin{align*} 
\exp(-\rw\{\Loop \textnormal{ visits } e \textnormal{ and not } -e\})= 1-\lambda( \Loops^1_e) .
\end{align*}
In particular, the left-hand is linear in $x_{\ue}$.\end{lemma}
\begin{proof}
Let $\Loops_e$ be the set of loops which visit $e$ and do not visit $-e$, and let $\Loops^*_e \subset \Loops_e$ be the set of loops {rooted at} $e$. Let $\Loops^{\circ}_e$ be the set of 
unrooted loops which have a representative in~$\Loops^*_e$. Note that $\Loops_e$ is the set of all representatives of the unrooted loops from $\Loops^{\circ}_e$.

Let $k_{\Loop}$ be the number of times  $\Loop$ visits $e$.
Observe that the number of all representatives of $\uLoop \in \Loops^{\circ}_e$ is $|\uLoop|/m_{\uLoop}$, and the number of its representatives in $\Loops^*_e$ is
${k_{\uLoop}}/{m_{\uLoop}}$.
Grouping the loops by their unrooted versions, the negated logarithm of the left-hand side of the desired equality becomes
\begin{align*}
\sum_{\Loop\in \Loops_e}\frac{\lambda(\Loop)}{|\Loop|} = \sum_{ \Loop \in \Loops_e } \frac{\lambda(\Loop)}{m_{\Loop}}\frac{m_{\Loop}}{|\Loop|} =
 \sum_{ \uLoop \in \Loops^{\circ}_e } \frac{\lambda(\uLoop)}{m_{\uLoop}} = \sum_{ \Loop \in \Loops^*_e } \frac{\lambda(\Loop)}{m_{\Loop}}\frac{m_{\Loop}}{k_{\Loop}}  =
\sum_{\Loop \in \Loops^*_e}\frac{\lambda(\Loop)} {k_{\Loop}}.
\end{align*}
Note that each $\Loop \in \Loops^*_e$ has a unique representation $\Loop =\Loop^1 \oplus\Loop^2 \oplus \cdots \oplus \Loop^k$,
where $ \Loop^i \in \Loops^1_e$. It follows that $\Loops^*_e$ is a disjoint union of $(\Loops^1_e)^{\oplus k}$ taken over all~$k$.
Therefore, by multiplicativity,
\begin{align}
  \sum_{\Loop \in \Loops^*_e}\frac{\lambda(\Loop)}{k_{\Loop}} =  \sum_{k=1}^{\infty}\sum_{\Loop \in (\Loops^1_e)^{\oplus k}} \frac{\lambda(\Loop)}k =\nonumber \sum_{k=1}^{\infty}\frac{\lambda (\Loops^1_e)^k}{k} & = 
  -\ln(1-\lambda( \Loops^1_e)). \qedhere
\end{align}
\end{proof}

The next lemma is the only place where property~\eqref{prop2} is used.

\begin{lemma}[Specific cancellations] \label{lem:specific} For any directed edge $e$,
\[
\rw\{\Loop \textnormal{ visits } e \textnormal{ and }  -e\}=0.
\]
\end{lemma}
\begin{proof}
Take $\Loop$ which visits both $e$ and $-e$, and let~$l$ be the smallest index such that $\Loop_l=e$ or $\Loop_l=-e$.
Let $m$ be the largest index such that $\Loop_m=-\Loop_l$. 
Consider the loop $\Loop'=\Loop_{1,l} \oplus (\Loop_{l,m})^{-1} \oplus \Loop_{m,|\Loop|+1}$.
From multiplicativity and property~\eqref{prop2} it follows
that ${w(\Loop')}=-{w(\Loop)}$.
It is now enough to notice that the map $\Loop \mapsto \Loop'$ is an involution of the set of loops which visit both $e$ and~$-e$.
\end{proof}

A set $C\subset E$ is called a \emph{cycle} if each vertex of the unique non-trivial connected component of $(V,C)$ has exactly two neighbors.
We first prove the main theorem in the case when $\FinG$ is \emph{trivalent}, by which we mean that all vertices of $\FinG$ have at most three neighbors. 
The only property of trivalent graphs used here is that their even subgraphs are collections of disjoint cycles (see Figure~\ref{fig}). 
The general case is then reduced to the trivalent one by 
a vertex decoration method.

\subsubsection*{The trivalent case}
By Lemma~\ref{lem:loopexpansion}, for any directed edge $e$,
\begin{align*}
\det(\Id -\Lambda) = A\exp(-\rw\{\Loop \text{ visits } e \text{ or } -e\}),
\end{align*}
where $A$ does not depend on $x_{\ue}$. By Lemma~\ref{lem:specific} and property~\eqref{prop3}, the signed measure of the set of loops which visit $e$ or $-e$ equals
\begin{align*}
&\rw\{ \Loop \textnormal{ visits } e \textnormal{ and not } -e \} +\rw\{ \Loop \textnormal{ visits } -e \textnormal{ and not } e \} \\&+\rw\{\Loop \textnormal{ visits } e \text{ and }  -e\} =2\rw\{ \Loop \textnormal{ visits } e \textnormal{ and not } -e \}.\nonumber
\end{align*}
 Using Lemma~\ref{lem:generic} we conclude that ${\det}(\Id -\Lambda)$ is a square of a linear polynomial in $x_{\ue}$. Since this holds for each undirected edge $\ue$,
  ${\det}(\Id -\Lambda)$ is a square of a multi-linear polynomial in $x$.
 
It is now enough to prove that the coefficients of the polynomials given by the square roots of the left- and right-hand side of \eqref{eq:KacWard} are equal.
Note that to each cycle $C$ containing $n$ edges, there naturally correspond $2n$ self-avoiding loops of length $n$
which traverse the cycle (including two orientations). 
We will write $\polsim$ if the monomials which are multi-linear in $x$ are the same in both expressions.
Using Lemma~\ref{lem:loopexpansion} and property~\eqref{prop4}, we have
\begin{align*}
{\det}^{\frac12}&(\Id -\Lambda)= \exp\Big(-\sum_{\Loop } \frac{\rw(\Loop)}2\Big) \polsim \exp\Big(-\sum_{\Loop \textnormal{ self-av.}} \frac{\rw(\Loop)}2\Big)
\\&= \exp\Big(\sum_{\Loop \textnormal{ self-av.}} \frac{x(\Loop)}{2|\Loop|}\Big)
= \exp\Big(\sum_{C \textnormal{ cycle }}\prod_{\ue \in C} x_{\ue}\Big) \\ & \polsim  \sum_{k=0}^{\infty} \mathop{\sum_{\{C^1,\ldots,C^k\}}}_{C^i \textnormal{ disjoint }}\prod_{\ue \in \bigcup_{i=1}^kC^i} x_{\ue} =
\sum_{H \text{ even} } \hspace{0.1cm} \prod_{\ue \in H} x_{\ue}.
\end{align*}
Note that the sum in the last line is taken over unordered collections of cycles so the factor $1/k!$ coming from the exponential cancels out.

\subsubsection*{The non-trivalent case}
The idea is to construct a trivalent graph $\FinG^{\dagger}$
which has the same generating functions of even subgraphs and loops as $\FinG$.
To this end, take a vertex $v\in V$ of degree $k > 3$. Let $u^1,u^2,\ldots,u^k$ be a clockwise ordering of the neighbours of $v$.
Consider a decoration of $\FinG$ where $v$ is replaced by $k$ new vertices $v^1,v^2,\ldots,v^k$, the edges $\{u^i,v \}$ are replaced by new edges $\{u^i,v^i\}$,
and new edges $\{v^i,v^{i+1}\}$ are added for $i=1,2,\ldots, k-1$  (see Figure~\ref{fig:vertexmap}). Note that the edge $\{ v^k, v^1\}$ is not added. The edges $\{ u^i, v^i\}$ inherit the weight from $\{u^i,v\}$ and all the
edges $\{v^i,v^{i+1}\}$ get weight $1$. If one repeats this procedure for every vertex with more than three neighbors, one obtains a trivalent graph $\FinG^{\dagger}$. Note that there is a bijection between the edges of $\FinG$ and the edges of $\FinG^{\dagger}$ which inherited the weights from $\FinG$.

It is easy to see that there is a weight-preserving bijection between the even subgraphs of $\FinG$ and $\FinG^{\dagger}$. For an even subgraph of $\FinG$, it is enough to take the corresponding edges in~$\FinG^{\dagger}$ and connect them in a unique way using the edges with weight $1$. Uniqueness is guaranteed by the construction of~$\FinG^{\dagger}$. There is also a weight-preserving bijection between the loops in $\FinG$ and~$\FinG^{\dagger}$. For a loop $\Loop$, we can construct the corresponding loop $\Loop^{\dagger}$ step by step.
If $\Loop$ makes a step from $(u^i,v)$ to $(v,u^j)$, then $\Loop^{\dagger}$ traverses the unique path starting at $(u^i,v^i)$, then following the edges of weight $1$, and ending at $(v^{j},u^j)$. 
It is clear that $x(\Loop)=x(\Loop^{\dagger})$, and one can check that $\alpha (\Loop)=\alpha (\Loop^{\dagger})$.

It is now enough to use Lemma~\ref{lem:loopexpansion} for $\FinG$, and pass to $\FinG^{\dagger}$ without changing the loop weights, then use the identity 
$\exp(-w(\Loops)) = \sum_{H \text{ even}} \prod_{\ue \in H} x_{\ue}$ for $\FinG^{\dagger}$, and go back to $\FinG$ in the even subgraph generating function.
\begin{figure}
		\begin{center}
			\includegraphics{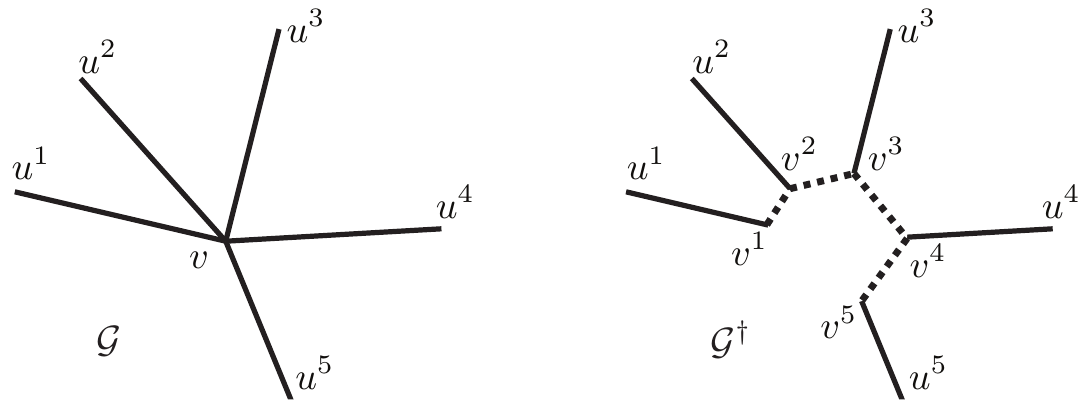}
		\end{center}
		\caption{The vertex decoration used to obtain a trivalent graph}
		\label{fig:vertexmap}
\end{figure}

\bibliographystyle{amsplain}
\bibliography{ising}
\end{document}